\documentclass[11pt]{article}
\usepackage[margin=0.75in]{geometry}
\usepackage{amssymb, amsthm,amsmath, bm, bbm, mathrsfs, graphicx, color}
\usepackage{ulem}

\usepackage{setspace}
\usepackage{latexsym}
\usepackage{multicol}

\newcommand{\I}{\mathbb{I}}

\newcommand\ph{\textcolor{black}}
\newcommand{\BL}{\textcolor{black}}

\newtheorem{lemma}{Lemma}

\newtheorem{remark}{Remark}

\newtheorem{proposition}{Proposition}

\begin{document}
\noindent Draft of \today
\onehalfspacing
\begin{center}

\textbf{\sc \Large  Exponentially Titled Empirical Distribution  Function }\\
\textbf{\sc \Large for Ranked Set Samples}

\vspace{0.5cm}

{Saeid Amiri$^{a,}$\footnote{Corresponding author: saeid.amiri1@gmail.com}, Mohammad Jafari Jozani$^{b}$,  and Reza Modarres$^c$ 
}

\vspace{0.3cm}

{\it $^a$ Department of Statistics, University of Nebraska-Lincoln, Lincoln, Nebraska, USA}

{\it $^b$ Department of Statistics,  University of Manitoba,
Winnipeg, MB,  CANADA, R3T 2N2

{\it $^c$ Department of Statistics, The George Washington University, Washington DC, USA }}
\end{center}

\begin{abstract}

We study nonparametric estimation of the distribution function (DF) of  a continuous random variable based on a ranked set sampling design   using the  exponentially tilted (ET) empirical likelihood method.   We propose    ET estimators of the DF and use them to construct new resampling algorithms for unbalanced ranked set samples.  We explore the properties of the proposed algorithms.  For a  hypothesis testing problem about the underlying population mean, we show that the bootstrap tests based on the ET estimators of the DF  are asymptotically normal and exhibit a small  bias of order $O(n^{-1})$. 
We illustrate the methods and evaluate the finite sample performance of the  algorithms  under both perfect and  imperfect ranking schemes using  a  real data set and several  Monte Carlo simulation studies. We compare the performance of the  test statistics based on the  ET  estimators  with those based on the empirical likelihood estimators.   \\
\end{abstract}


%%%%%%%%%%%%%%%%%%%%%%%%%%%%%%%%%%%%%%%%%%%%%%%%%%%
%%
%%         Section 1: Introduction
%%
%%%%%%%%%%%%%%%%%%%%%%%%%%%%%%%%%%%%%%%%%%%%%%%%%%%%

\section{Introduction}

Ranked set   sampling (RSS)  is a powerful and  cost-effective  data collection  technique that is often used to collect more representative samples from the underlying population when  a small number of sampling units can be fairly accurately ordered  without taking actual measurements on the variable of interest.   RSS is most effective when obtaining exact measurement on the variable of interest is very costly,   but  ranking the sampling units  is relatively  inexpensive. 
RSS finds applications  in  industrial statistics, 
environmental and ecological studies as well as  medical sciences.
For recent overviews of the theory and applications of RSS and  its variations see Wolfe (2012) and Chen et al.\ (2004).

Ranked set samples can be either balanced or unbalanced. An unbalanced ranked set sample (URSS) is one in which the ranked order statistics are not quantified the same number of times.  To obtain an URSS of size $n$ from the   underlying population  we proceed as follows. Let $n$ sets of sampling units, each of size $k$, be  randomly chosen from the population using  a simple random sampling  (SRS) technique. The units of each set are ranked  by any means other than the actual quantification of the variable of interest. Finally, one and only one unit in each ordered set  with a pre-specified rank is measured.
Let $m_r$ be the number of measurements on units with rank $r$,  $r\in\{1, \ldots, k\}$  such that $n= \sum_{r=1}^k m_r$. Suppose  $X_{(r)j}$ denotes the measurement on the $j$th unit with rank $r$. The resulting  URSS of size $n$ from the underlying population is denoted by    ${\bf X}_{URSS}=\{\mathcal{X}_1, \ldots, \mathcal{X}_n\}$, where 
the elements of the $r$th row $\mathcal{X}_r= (X_{(r)1}, X_{(r)2}, \ldots, X_{(r)m_{r}})$ are independently and identically distributed (i.i.d.) from  $F_{(r)},  r=1, \ldots, k$ and  $F_{(r)}$ is the DF  of the $r$th order statistic. \ph{ Moreover, 
$X_{(r)j}$s are independent  for $r=1, \ldots, k$ and $j=1, \ldots, m_r$.}  Note that if $m_r=m$, $r=1, \ldots, k$, then URSS reduces to  the balanced  RSS. The DF of URSS is  
\begin{eqnarray}\widehat F_{q_n}(t)=\frac{1}{n}\sum_{r=1}^{k}
\sum_{j=1}^{m_r}I(X_{(r)j}\leq t)=\sum_{r=1}^k q_{m_r}\widehat F_{(r)}(t),
\label{rsss3}
\end{eqnarray}
where $n=\sum m_r$ and $q_{m_r}=m_r/n$.   As it is shown in Chen et al.\ (2004), 
when $n\longrightarrow \infty$, and $q_{m_r}\longrightarrow q_r$, for  $r=1,\ldots,k$, we have  $\widehat F_{q_n}(t)\longrightarrow F_q(t)$,  where
\begin{eqnarray}
  F_{q}(t)=\sum_{r=1}^k q_{r} F_{(r)}(t). \label{rss7}
\end{eqnarray}
One can easily see that $F_{q}(t)$ is not equal to the underlying DF $F(t)$,  unless   $q_r= 1/k$, $r=1, \ldots, k$, showing that the EDF based on the URSS data   does not provide a  good  estimate of  the underlying distribution $F$.  The properties of the   EDF of  the balanced and unbalanced RSS are  studied in Stokes and Sager (1988) as well as Chen et al.\ (2004).  

\BL{ In this paper, we  use the empirical likelihood method as a nonparametric approach for estimating $F$.  To this end, we  propose two methods to estimate $F$ using the  exponentially tilted (ET)   technique. The proposed estimators  can be used   as standard tools for practitioners to estimate the  standard error of any well-defined statistic based on RSS or URSS data  and  to make  inferences about the characteristics of interest of the underlying population.   Another interesting problem in this direction is to develop efficient resampling techniques for  URSS data,  as in many cases the exact or the asymptotic distribution of the statistics based on URSS data   are  not available or they are  very  difficult to obtain (e.g., Chen et al., 2004). Akin to the methods of  Modarres et al.\ (2006) and Amiri et al.\ (2014), the new  ET estimators  of $F$  are used to construct new  resampling  techniques for URSS data.  We study different properties of the proposed algorithms. For a hypothesis  testing  problem, about the underlying population mean,  we show that  the bootstrap tests based  on  the ET estimators  are asymptotically normal and exhibit a small bias of order $O(n^{-1})$ which are  desirable properties.  }

The outline of the paper is as follows. In  Section \ref{ET},  we present  ET estimators of $F$   based on the URSS data. 
 Section \ref{BT-URSS} considers two  methods  for resampling  RSS and URSS data based on the ET estimators of $F$. We provide justifications for validity of these methods for a hypothesis  testing problem about the population mean. 
 Section \ref{simulation} describes a simulation study to compare the finite sampling properties of the proposed methods with parametric bootstrap and some existing resampling techniques for testing a hypothesis about the population mean. We consider both perfect and  imperfect  ranking  scenarios,  three different  distributions and five RSS designs.  We compare the performance of  our proposed methods with the one based on the empirical likelihood method studied in Liu et al. (2009) as well as Baklizi (2009).  In Section \ref{real data}, we apply our methods for a testing hypothesis problem using a real data set  consisting of the birth weight and seven-month weight of  224 lambs along with
the mother's weight at time of mating.   Section \ref{concluding}  provides some concluding remarks.  

%%%%%%%%%%%%%%%%%%%%%%%%%%%%%%%%%%%%%%%%%%%%%%%%%%%
%%
%%         Section 2: ET method
%%
%%%%%%%%%%%%%%%%%%%%%%%%%%%%%%%%%%%%%%%%%%%%%%%%%%%%

\section{ {Exponential Tilting of DF}}\label{ET}
Exponential tilting of an empirical likelihood is a powerful technique in nonparametric statistical inference. 
The impetus of this approach is the use of the estimated DF subject to some constraints rather than the EDF.  
ET methods find applications in computation of bootstrap tail 
probabilities (Efron and Tibshirani, 1993), point estimation (Schennach, 2007),
estimation of the spatial quantile regression (Kostov, 2012), Bayesian treatment of quantile
regression (Schennach, 2005),  small area estimation (Chaudhuri and Ghosh, 2011) and Calibration estimation (Kim, 2010), among others.

Let ${\bf X}= \{ X_1, \ldots, X_n\}$ be a generic sample of size $n$ from $F$  and suppose  $F_n(x)=\sum_{i=1}^n \frac{1}{n} \I(X_i\leq x)$  is   the EDF of ${\bf X}$ which places empirical frequencies (weights) $1/n$ on each $X_i$.  Consider an estimator $ \widetilde F_{p}(x)=\sum_{i=1}^n p_i\I(X_i\leq x)$ of $F$ which assigns weights $ p_i$ instead of $1/n$ to  each $X_i$. To obtain the ET estimator of $F$,  we  minimize an aggregated  distance between the empirical weights $1/n$ and $p_i$ subject to some constraints on the $p_i$'s. More specifically,  one chooses a distance  $d(\widetilde F_{ p},F_n)=\sum_{i=1}^{n}
d(p_{i},\frac 1n)$ and minimizes  $d(\widetilde  F_{ p},F_n)$ subject to $\sum_{i=1}^n p_i=1$ and some other constrains such as $ g({\bf X}, \theta_0)= \sum_{i=1}^n p_i g(X_i, \theta_0)=0$, using the following Lagrangian multiplier method
\begin{eqnarray} \label{lagrange}
d(\widetilde F_{p},F_{n})-\lambda
g({\bf X},\theta_0)-\alpha \big(\sum_{i=1}^{n} p_{i}-1\big),
\end{eqnarray} 
 where $g({\bf X}, \theta_0)$ is often imposed under the null hypothesis in a testing problem  or any other  conditions that one needs to account for in practice. Note that the minimization in \eqref{lagrange} can also be done by minimizing the distance between $\widetilde  F_{ p}(x)$ and any target estimator $F_{\widehat p}(x)=\sum_{i=1}^n\widehat{p}_i\I(X_i\leq x)$ other than the EDF $F_n(t)$.

The choice of the discrepancy function $d(\cdot,\cdot)$ for the aggregated loss $d(\widetilde F_p,F_n)$  in \eqref{lagrange} leads to different ET estimators of  $F$. Since  $F_n(x)$ is the  nonparametric maximum likelihood estimator of $F$ under the Kullback-Leibler  distance subject to the restriction $\sum_{i=1}^n p_i=1$,  one often uses
$$d(\widetilde F_{ p},F_{\widehat p})= \sum_{i=1}^n p_i \log \left(\frac{p_i}{\widehat p_i} \right).$$
We   propose two ET estimators of $F$ based on  URSS data  with sample size $n=\sum_{r=1}^k m_r$ where $k$ is the set size. \BL{The ET estimators are then used to propose new bootstrapping algorithms from URSS data}.

%%%%%%%%%%%%%%%%%%%%
%
%  ET by All Observations
%
%%%%%%%%%%%%%%%%%%%%%%%%%

\subsection{ Exponential Tilting of  All  Observations (EAT)}
\BL{In this section, we propose  our first  ET estimator of $F$  which is later used to  resample   from  within each row  of ${\bf X}_{URSS}=\{ X_{(r)j}, r=1, \ldots, k; j=1, \ldots, m_r\}$. 
%Suppose we have an  unbalanced ranked set sample of size $n$  from   $F$. 
 The idea behind the first ET estimator of $F$, for bootstrapping  ${\bf X}_{URSS}$, is to find an estimator 
  \begin{align}\label{first}\widetilde F_{p}(x)= \sum_{r=1}^k \sum_{j=1}^{m_r} p_{(r)j} \I(X_{(r)j}\leq x),
  \end{align}
subject to the constraints   
\begin{align}\label{const1}
\sum_{r=1}^k \sum_{j=1}^{m_r} p_{(r)j}=1\quad \text{and}\quad \sum_{r=1}^{k} \sum_{j=1}^{m_r} p_{(r)j} X_{(r)j}=\overline{{X}}_{URSS},
\end{align}
 where $\overline{X}_{URSS}=\displaystyle\frac{1}{n}\sum_{r=1}^{k}\sum_{j=1}^{m_r} X_{(r)j}$.
  \begin{lemma} \label{lem1}
  Let ${\bf X}_{URSS}=\{ X_{(r)j}, r=1, \ldots, k; j=1, \ldots, m_r\}$ be a URSS sample of size $n$ from the underlying population $F$ when the set size is $k$ and $X_{(r)j }\in {\bf R}$ is the $r$-th order statistic in a simple random sample of size $k$ from $F$. The optimum values of $p_{(r)j}$ in \eqref{first} under the constraints \eqref{const1}  are given by 
  \begin{eqnarray}
  \tilde p_{(r)j}=\frac{\exp(\lambda X_{(r)j})}{\displaystyle\sum_{r=1}^k\sum_{j=1}^{m_r} \exp(\lambda X_{(r)j})}, \quad r=1,\ldots, k;~ j=1, \ldots, m_r,
  \label{eta1}
  \end{eqnarray}
  where $\lambda$ is obtained from $\sum_{r=1}^k \sum_{j=1}^{m_r} p_{(r)j} X_{(r)j}=\overline{X}_{URSS}$.
  \end{lemma}  
  \begin{proof}
Using the   Lagrange multipliers method,  and by   minimizing 
\begin{eqnarray}
\sum_{r=1}^k \sum_{j=1}^{m_r} p_{(r)j}\ln \left( \frac{p_{(r)j}}{1/n}\right)+\lambda (\sum_{r=1}^k \sum_{j=1}^{m_r} p_{(r)j} X_{(r)j}-\overline{X}_{URSS})+\alpha(\sum_{r=1}^k \sum_{j=1}^{m_r} p_{(r)j}-1),
\label{LA}
\end{eqnarray}
with respect to $p_{(r)j}$'s,  one can easily obtain    the optimum values in \eqref{eta1}.
  \end{proof}
 In Section \ref{BT-URSS},  we use  $\tilde {F}_p(x)=\sum_{r=1}^k \sum_{j=1}^{m_r} \tilde p_{(r)j} \I(X_{(r)j}\leq x)$ for bootstrapping $X_{URSS}$ instead of the commonly used empirical DF.  It is worth noting that for hypothesis testing  problems about the underlying population mean $\mu$ involving the null hypothesis   $H_0: \mu=\mu_0$, minimization in \eqref{LA} is done subject to the condition $\sum_{r=1}^{k} \sum_{j=1}^{m_k} p_{(r)j} X_{(r)j}=\mu_0$. Using the optimum weights $\widetilde{p}_{(r)j}$ from the ET estimate of $F$, we also propose
$S^2=\sum_{r=1}^k \sum_{j=1}^{m_r} \tilde p_{(r)j} (X_{(r)j}-\overline{X}_{URSS})^2$
to estimate the population variance $\sigma^2$.  
}
%%%%%%%%%%%%%%%%%%%%
%
%  ET by Row Observations
%
%%%%%%%%%%%%%%%%%%%%%%%%%

\subsection{Exponential Tilting  of    Rows  (EAR)}
By  the structure of the URSS data,  ${\bf X}_{URSS}$,    we observe that  $X_{(r)1}, \ldots, X_{(r)m_r}$ are i.i.d.\  samples from  $F_{(r)}(\cdot)$, which is the distribution of the $r$-th order statistic \BL{in a simple random sample of size $k$ from $F$. Since 
$$F(t)=\frac{1}{k} \sum_{r=1}^k  F_{(r)}(t), $$
the idea  behind our  next proposed ET  estimator of $F$ is  to estimate each $F_{(r)}$ using $X_{(r)1}, \ldots, X_{(r)m_r}$,} and construct an estimator of $F$ by averaging over these estimators using suitable weights obtained from the Lagrange multipliers method under some constraints. To this end, we work with an estimator of $F$ of the form  
\begin{align}\label{est2}
\widetilde F_{p}(t) = \sum_{r=1}^k p_{(r)} \widehat{F}_{(r)}(t), 
\end{align}
where $\widehat{F}_{(r)}(t)= \frac{1}{m_r} \sum_{j=1}^{m_r} \I(X_{(r)j}\leq t)$ is the EDF of $X_{(r)1}, \ldots, X_{(r)m_r}$. 
 \BL{
  \begin{lemma} \label{lem2}
  Let ${\bf X}_{URSS}=\{ X_{(r)j}, r=1, \ldots, k; j=1, \ldots, m_r\}$ be a URSS sample of size $n$ from $F$ where the set size is $k$  and $\{X_{(r)j}, j=1, \ldots, m_r \}$ are i.i.d.\ samples from $F_{(r)}$ the DF of the $r$-th order statistic of a simple random sample of size $k$ from $F$. Then,  an optimum estimator of $F$ in the form of \eqref{est2} under the constraints   $\sum_{r=1}^{k}p_{(r)}=1$ and $\sum_{r=1}^k p_{(r)} \bar{X}_{(r)}= \overline X_{URSS}$, where $\overline{X}_{(r)}=\frac{1}{m_r}\sum_{j=1}^{m_r} X_{(r)j}$, $r=1, \ldots, k$,  is given by
   \begin{eqnarray}
\tilde {F}_p(x)=\displaystyle\sum_{r=1}^k   \frac{\tilde p_{(r)}}{m_r} \sum_{j=1}^{m_r} \I(X_{(r)j}\leq x)\quad\text{with}\quad   \tilde p_{(r)}=\frac{\exp(\lambda \overline X_{(r)})}{\sum_{r=1}^k \exp(\lambda \overline X_{(r)})}, 
  \label{eta2}
  \end{eqnarray}
  where  $\lambda$ is obtained  from $\sum_{r=1}^k  p_{(r)} \overline X_{(r)}=\overline{X}_{URSS}$. 
  \end{lemma}  
  \begin{proof} 
The results easily follow using the  Lagrange multipliers method and  minimizing
\begin{eqnarray}
\sum_{r=1}^k  p_{(r)}\ln \left( \frac{p_{(r)}}{1/k}\right)+\lambda (\sum_{r=1}^k  p_{(r)} \overline X_{(r)}-\overline{X}_{URSS})+\alpha(\sum_{r=1}^k p_{(r)}-1).
\label{LA1}
\end{eqnarray}
with respect to  $p_{(r)}$. 
   \end{proof}
 In Section \ref{BT-URSS}, we use  $\tilde {F}_p(x)=\displaystyle\sum_{r=1}^k   \frac{\tilde p_{(r)}}{m_r} \sum_{j=1}^{m_r} \I(X_{(r)j}\leq x)$  and propose a new bootstrapping algorithm to  resample from $X_{URSS}$ instead of the commonly used empirical DF.   Here again for  hypothesis testing problems involving  $H_0: \mu=\mu_0$ where $\mu$ is the population mean, minimization in  \eqref{LA1} is done subject to the condition $\displaystyle\sum_{r=1}^k  p_{(r)} \overline X_{(r)}=\mu_0$. }
\begin{remark}
If for  the observed URSS data all the  $m_r$s are large enough, then one can use ET estimators of $F_{(r)}$ by simply treating $X_{(r)j}$'s as a SRS of size $m_r$ from $F_{(r)}$ and constructing   the estimator $\widetilde{F}(t)= \frac{1}{k} \sum_{r=1}^k \widetilde F_{(r)}(t)$ for $F$. Here,   $\widetilde{F}_{(r)}(t)= \sum_{j=1}^{m_r} w_j(r) \I(X_{(r)j}\leq t)$ and $w_j(r)$s are  obtained subject to constraints $\sum_{j=1}^{m_r}w_j(r)=1$ and $\sum_{j=1}^{m_r}w_j(r) X_{(r)j}= \overline{X}_{(r)}$, for $r=1, \ldots, k$, using the following  Lagrange multipliers problems:
$$\sum_{j=1}^{m_r}  w_{j}(r)\ln \left( \frac{w_j(r)}{1/m_r}\right)+\lambda_r (\sum_{j=1}^{m_r}  w_{j}(r) \overline X_{(r)j}-\overline{X}_{(r)})+\alpha_r (\sum_{j=1}^{m_r} w_{j}(r)-1), \quad r=1, \ldots, k. $$
\end{remark}

%%%%%%%%%%%%%%%%%%%%%%%%%%
%
%
%  Bootstrapping URSS\RSS
%
%%%%%%%%%%%%%%%%%%%%%%%%%%%

\section{Bootstrapping URSS and  RSS}\label{BT-URSS}
\BL{ In this section, we propose  two new bootstrapping techniques to resample from a  balanced or unbalanced ranked set sample of size $n$. %Two  
 The first algorithm is based on the ET estimator of $F$  in Lemma \ref{lem1} to resample  the entire URSS while the second one uses  the ET estimator of $F$ in Lemma \ref{lem2}  to   resample from  within each row separately. }
 We note that most of  the bootstrap methods developed for RSS are   based on the EDF and one can easily modify them  using  ET estimators of $F$. Monte Carlo simulation studies  indicate that  bootstrapping methods based on the ET estimators of $F$ perform better than their counterparts using the  EDF.  

%%%%%%%%%%%%%%%%%%%%%%%%%%
%
%
%  Bootstrapping URSS\RSS with EAT
%
%%%%%%%%%%%%%%%%%%%%%%%%%%%

\subsection{Bootstrapping Algorithm: EAT}

To  resample  from  the ET estimator of $F$ given by  $$\tilde {F}_p(x)=\displaystyle\sum_{r=1}^k \sum_{j=1}^{m_r} \tilde p_{(r)j} \I(X_{(r)j}\leq x),$$ where $\tilde{p}_{(r)j}$ is defined in \eqref{eta1} we proceed as follows:
\begin{itemize}
\item [1.] Assign probability $\tilde p_{(r)j}$ to each element $X_{(r)j}$ of ${\bf X}_{URSS}$. 
\item[2.] Randomly draw $X_1^\diamond,\ldots ,X_k^\diamond$ from ${\bf X}_{URSS}$ according to probabilities  $\{\tilde p_{(r)j}\}$, order them as $X_{(1)}^\diamond \leq\ldots \leq X_{(k)}^\diamond$ and retain $X^{*}_{(r)1}=X_{(r)}^\diamond$.
\item[3.] Repeat Step 2,  for $r=1,\ldots ,k$ and $j=1,\ldots ,m_r$ to generate a   bootstrap URSS $\left\{X^{*}_{(r)j}\right \}$.
\item[4.] Repeat  steps 2-3, $B$ times to obtain the bootstrap samples.
\end{itemize}

\BL{
One can easily validate the use of the ET estimator of $F$ for different bootstrapping purposes. For example, suppose we want to carry out a bootstrap  test for testing  $H_0: \mu=\mu_0$ against $H_a: \mu>\mu_0$,  where $\mu$ is the unknown parameter of interest.  Using  Hall (1992),    the Edgeworth expansion of the $p$-value for testing $H_0$ against $H_a$  based on a SRS of size $mk$ from the underlying population   with the test statistic    $T= \frac{\bar X-\mu_0}{S/\sqrt{mk}}$,  is given by
\begin{eqnarray}\label{pvalue-srs}
P=P(T \geq  t)=1-\Phi(t)-(mk)^{-1/2} q(t)\phi(t)+O(\frac{1}{mk}), 
\end{eqnarray}
where $q(\cdot)$ is a  quadratic function and  $\Phi(\cdot)$ and $\phi(\cdot)$ are the standard normal distribution and density functions,  respectively. We consider the problem for a balanced RSS case,  as the following argument  can also be applied to URSS data with some modifications.
Let $\{ X_{(r)j},  r=1, \ldots, k;~ j=1, \ldots, m\}$ be a balanced ranked set sample of size $mk$ from the underlying population with mean $\mu$.  We show that  the ET bootstrap approximation of  the sampling distribution of $T$ is in error by only $1/mk$ and    the $p$-value obtained through the EAT method has the desirable second order accuracy  
This is similar to results obtained in DiCiccio and Romano (1990). For more details see Efron (1981) and Feuerveger et al.\ (1999).  
\begin{proposition}
Suppose  $\{ X^*_{(r)j}, r=1, \ldots, k;~ j=1, \ldots, m\} $ is a bootstrap sample generated  from the EAT algorithm. Let  $T^*= \frac{(\bar X^*-\bar X)}{S^{*}/\sqrt{mk}}$ be the bootstrap test  for testing $H_0: \mu=\mu_0$  with p-value $P^*$,  where  $\bar X^*$  is the mean of the bootstrap sample obtained form the ET estimator of $F$ and $S^{2*}=\frac 1k \sum_{r=1}^k S^{2*}_{(r)}$ with $S^{2*}_{(r)}= \frac{1}{m-1}\sum_{j=1}^m(X^*_{(r)j}-\bar{X}^*_{(r)})^2$. 
Then, 
\begin{eqnarray}
P-P^*=O(\frac{1}{mk}),
\end{eqnarray} 
where $P$,   given by \eqref{pvalue-srs}, is the p-value  of the usual  $T$-test based on a simple random sample of comparable size $mk$ from the underlying population. 
\end{proposition}
\begin{proof}
For simplicity, we write the resampled data as 
$\{X_{1}^*,\ldots,X_{km}^*\}$. In order to test $H_0: \mu=\mu_0$, and  to ensure that the null hypothesis is incorporated into the  ET estimator of $F$, we introduce the Lagrange multipliers for the constraints $\sum_{i=1}^n \tilde p_i=1$  and $\sum_{i=1}^n \tilde p_iX^*_i=\mu_0$, where the weights $\tilde p_i$ are obtained as  
\begin{eqnarray}
\tilde p_{i}(\mu_0)=\frac{\exp(\lambda(\mu_0) X^*_{i})}{\sum_{j=1}^{mk} \exp(\lambda(\mu_0) X^*_{j})}, ~ i=1,\ldots,km,
\end{eqnarray}
and $\lambda(\mu_0)$ is the coefficient calculated from  $\sum_{i=1}^n \tilde p_i(\mu_0) X^*_i=\mu_0$. One can easily show that  $X^*_i$s are generated  from   
\begin{eqnarray}
dF_p(x)=e^{\{A(\lambda (\mu))-\lambda(\mu) x\} }dF_n(x),
\end{eqnarray}
where 
 $A(\lambda(\mu))=\log(\frac{1}{mk} \sum_{i=1}^k exp(\lambda(\mu)  X_i))$.  To obtain the ET estimator of $F$ under the null hypothesis  we must  have  
\[
\mu_0=A'(\lambda(\mu_0))=\frac{\sum_{i=1}^{mk} X_i exp(\lambda(\mu_0) X_i)}{\sum_{i=1}^{mk} exp(\lambda(\mu_0) X_i)}.
\]
Therefore, one can use the bootstrap  test statistic $T^*= \frac{(\bar X^*-\bar X)}{S^{*}/\sqrt{mk}}$  for testing $H_0: \mu=\mu_0$  where  $\bar X^*$  is the mean of the bootstrap sample obtained form the ET estimator of $F$ and $S^{2*}=\frac 1k \sum_{r=1}^k S^{2*}_{(r)}$ with $S^{2*}_{(r)}= \frac{1}{m-1}\sum_{j=1}^m(X^*_{(r)j}-\bar{X}^*_{(r)})^2$.   Following Hall (1992) and using the Edgeworth expansion, the $p$-value for testing  $H_0: \mu=\mu_0$ against $H_0:\mu >\mu_0$ using the  bootstrap test statistic $T^*$  is given by 
\begin{eqnarray}
P^*=P(T^* \geq t|F_p)=1-\Phi(t)-(mk)^{-1/2} \widehat q(t)\phi(t)+O(\frac{1}{mk}), \nonumber%\label{conv1}
\end{eqnarray}
where $\widehat q$ is a  quadratic function. Now, the results follows from \eqref{pvalue-srs}.
%Also, using  Hall (1992),    the Edgewoth expansion of the $p$-value  based on a SRS of size $mk$ from the underlying population  for testing $H_0: \mu=\mu_0$  with the test statistic    $T= \frac{\bar X-\mu_0}{S/\sqrt{mk}}$,  leads to
%\begin{eqnarray}
%P=P(T \geq  t)=1-\Phi(t)-(mk)^{-1/2} q(t)\phi(t)+O(\frac{1}{mk}). \nonumber
%\end{eqnarray}
%Consequently,
%\begin{eqnarray}
%P-P^*=O(\frac{1}{mk}), \nonumber
%\end{eqnarray} 
\end{proof}
}
%which shows  that the ET bootstrap approximation of  the sampling distribution of $T$ is in error by only $1/mk$ and    the $p$-value obtained through the EAT method has the desirable second order accuracy. This is similar to results obtained in DiCiccio and Romano (1990). For more details we refer to Efron (1981) and Feuerveger et al.\ (1999).  
%
%%%%%%%%%%%%%%%%%%%%%%%%%%
%
%
%  Bootstrapping URSS\RSS with EAR
%
%%%%%%%%%%%%%%%%%%%%%%%%%%%

\subsection{Bootstrapping  Algorithm EAR}
The idea behind this  method is to use the ET estimator of $F$ given by  
$$\tilde {F}_p(x)=\displaystyle\sum_{r=1}^k   \frac{\tilde p_{(r)}}{m_r} \sum_{j=1}^{m_r} \I(X_{(r)j}\leq x), $$
where $\tilde p_{(r)}$ is defined in \eqref{eta2}. To this end we proceed as follows:
\begin{itemize}
\item[1.] Assign  probabilities $\tilde p_{(r)}$ to each row $\mathcal{X}_r$ of ${\bf X}_{URSS}$, $r=1, \ldots, k$. 
\item[2.] Select a row randomly using $\tilde p_{(r)}$ and select an observation randomly from that row. 
\item[3.] Continue step 2 for $k$ times to obtain $k$ observations. 
\item[4.] Order them as $X_{(1)}^\diamond\leq\ldots \leq X_{(k)}^\diamond$ and retain $X^{*}_{(r)1}=X_{(r)}^\diamond$
\item[5.] Perform Steps 2--4 for $m_r$ and obtain $\{X^{*}_{(r)1},\ldots,X^{*}_{(r)m_r}\}$.
\item[6.] Perform Steps 2--5 for $r=1,\ldots ,k$.
\item[7.] Repeat  steps 2--6, $B$ times to obtain the bootstrap samples.
\end{itemize}
\section{Monte Carlo Study}\label{simulation}
\BL{In this section, we compare the finite sample performance of  out nonparametric EAT and EAR resampling methods  with a parametric bootstrap (PB) procedure. The PB method uses a  parametric test (PT) with an  asymptotic normal distribution  to test the hypothesis $H_0: \mu=\mu_0$, where $\mu$ is the unknown parameter of interest and $\mu_0$ is a known constant.}   The resampling is performed using B=500 resamples and the entire experiment is then replicated 2000 times. We use  several RSS and URSS designs with different sample sizes  when  the set size is chosen to be $k=5$. We also conducted unreported simulation studies for other values of $k$ and we observed similar performance that we summarize below. 

 The  RSS designs that we consider  are written as  $D=(m_{1}, m_{2}, \ldots, m_{5})$ with $n_{D}=\sum_{r=1}^k m_r$. For example, the first design
is balanced with $k=5$ and $m_r=5$ observations per stratum, which is denoted by  $$D_1=(5,5,5,5,5)\quad \text{with} \quad n_{D_1}= 25.$$
 Similarly, we define the following designs,
\begin{align*}
&D_2=(8,3,3,2,4)\quad \text{   with}\quad   n_{D_2}= 25, \nonumber\\
&D_3=(3,2,5,8,3)\quad~  \text{with}\quad n_{D_3}= 21,\nonumber\\
&D_4=(3,10,3,3,3)\quad\text{with}\quad  n_{D_4}= 22,\nonumber\\
&D_5=(4,2,3,3,8)\quad~\text{with}\quad  n_{D_5}= 24.\nonumber
\end{align*}
We obtain samples from the Normal(0,1),  Logistic(1,1) and  Exponential(1) distributions.

%%%%%%%%%%%%%%%%%%%%%%%%%%%%%%%%%%%%%%%%%%%%%%%%%%%
%%
%%         Subsection :Testing the population mean
%%
%%%%%%%%%%%%%%%%%%%%%%%%%%%%%%%%%%%%%%%%%%%%%%%%%%%%
\subsection{ Testing a hypothesis about  the population  mean }
\BL{We first proceed with the following proposition.
 \begin{proposition} \label{p-bt1} Suppose $F$ is the DF of the variable of interest in the underlying population with $\int x^2dF(x)<\infty$. Let  $\widehat F_{(r)}$ be  the EDF of the $r^\text{th}$ row of a balanced RSS data  and $\mu$ represent the population mean.
 Then $(\vartheta_1,\ldots ,\vartheta_k)$, with $\vartheta_i= \mu (\widehat F_{(i)})-\mu (F_{(i)})$, converges in distribution to a multivariate normal distribution with  the mean vector zero and
the  covariance matrix $\Sigma=diag(\sigma^2(F_{(1)})/m,\ldots ,\sigma^2(F_{(k)})/m)$ where $\sigma^2(F_{(i)})=\int (X-\mu_{(i)})^2dF_{(i)}$ and $\mu_{(i)}=\int xdF_{(i)}(x)$.
 \end{proposition}
  This proposition suggests to use the following test statistic for testing the hypothesis $H_0: \mu=\mu_0$
\begin{eqnarray}
T(X,\mu_0)=\frac{1}{k}\sum_{r=1}^k \left(\frac{\bar X_{(r)}-\mu_0}{S}\right)\overset{d}{\rightarrow} N(0,1),
\label{urssa5}
\end{eqnarray}
where 
\begin{eqnarray}
S^2= \frac{1}{k^2}\sum_{r=1}^k\frac{S^2(X_{(r)})}{m_r}.
\label{urssa6}
\end{eqnarray}
  The test statistic  $T(X,\mu_0)$, which is approximately Normal$(0,1)$ for large $k$,  is referred to as the PT in the rest of the work.    Ahn et al. (2014) consider  the Welch-type (WT) approximation to the distribution $T(X,\mu_0)$, where  the degree of freedom of the test  can be approximated using    
\begin{eqnarray}
S^2= \Big (\sum_{r=1}^k\frac{S^2(X_{(r)})}{m_r} \Big)^2\Big/ \Big (\sum_{r=1}^k\frac{S^4(X_{(r)})}{m^2_r(m_r-1)} \Big).
\label{df}
\end{eqnarray}
The nonparametric bootstrap tests using the EAT and EAR  methods are conducted based on the following  steps: } 
\begin{itemize}
  \item[1.]Let X be an URSS/RSS sample from $F$.
  \item[2.]Calculate $T=T(X,\mu_0)$, given in (\ref{urssa5}), under the null hypothesis $H_0: \mu=\mu_0$.
  \item[4.]Apply each of the resampling procedures on $X$  to obtain $X_b^{*}=\{X^{*}_{(r)j}\}_b$.
\item[5.] Calculate $T_b^*=T(X_b^{*},\mu_0)$, $b=1,\ldots,B$.
  \item[6.]Obtain the proportion of rejections via $\frac{\#\{T_b^*>T\}}{B}$ to estimate the  $p$-value.
\end{itemize}
%\BLL{It must be mention the test statitics unde resample should be $T(X_b^{*},\mu_0)$ not $T_b^*=T(X_b^{*},\bar X)$ that done under regular bootstrap, because by resampling using $\tilde{p}_{(r)}$ we are getting sample under null hypothesis not the sample. } 
We also performed the desired testing hypothesis using  PB  by generating URSS samples  from Normal(0,1), Logistic(1,1) and  exponential(1) distributions.  To perform PB test we use the following steps (for more details on PB method see   Efron and Tibshirani (1993)):
 \begin{itemize}
  \item[1.] Let X be a URSS sample from a distribution  $F_{\theta}$ where $\theta$ is the unknown parameter and let $\mu=E_{\theta}(X)$.
  \item[2.] Calculate $T=T(X, \mu_0)$, under the null hypothesis $H_0: \mu=\mu_0$. 
 \item [3.] Estimate $\theta$ from X and take a URSS from $F_{\widehat \theta}$,  
$X_b^{*}=\{X^{*}_{(r)j}\}_b$. 
  \item[4.]Calculate $T^*_{b}=T^*_{b}(X^*_b, \mu_0)$. 
  \item[5.] Obtain the proportion of rejections via $\frac{\#\{T_{b}^*>T\}}{B}$ to estimate the $p$-value.
\end{itemize}
%
%\BLL{
%In order to achieve two sidede test can follow as singh (2011), he outlines the preceding to 
%$H_0: \theta= \theta_0$ vs $H_1: \theta\neq \theta_0$ is not so clear as the one sided hypothesis, he gave 
%\begin{eqnarray}
%\widehat p_B^*=2\min(\widehat p_B,1-\widehat p_B),
%\label{pvtwoa}
%\end{eqnarray}
%where $\widehat p_B$ is the one-sided testing problem, although if the test statistic is the pivotal, then 
%\begin{eqnarray}
%\widehat p_B= p_B\left(\left|\frac{\widehat \theta^*-\widehat \theta}{S(\widehat\theta^*)}\right|>\left|\frac{\widehat \theta-\theta_0}{S(\widehat\theta)}\right|\right),
%\label{pvtwob}
%\end{eqnarray}  
%which leads to 
%\begin{eqnarray}
%\widehat p_B= \frac{1}{B}\left( \sum_{b=1}^{B} I(T_b>|t|) + \sum_{b=1}^{B} I(T_b<-|t|)\right).
%\end{eqnarray}
%The main difference of (\ref{pvtwoa}) and (\ref{pvtwob}) is in the case of non-pivotal, the statement (\ref{pvtwob}) can not be used.
%}
To conduct the parametric bootstrap we estimated the population mean using  the sample mean  and used $\sigma=1$. Subsequently, we  generated samples from the N($\bar{x}$, 1), Logistic($\bar{x}$, 1) and Exponential($\bar{x}$) distributions.
% We used $B=500$ resamples from a given sample or simulated $B=500$ parametric bootstrap replications with $\alpha=0.05$. the entire simulation study is replicated 2000 times.
 % A $95\%$ confidence
%interval for $\alpha$ is $0.05 \pm 1.96 \sqrt{(0.05)(0.95)/500}=(0.031, 0.069)$.
%
%
%In order to compare the proposed methods, first a RSS sample using $D$ is generated, then some observations are excluded randomly to generate   $D_i$, $1=1, \ldots, 5$.
Table \ref{2observed_alpha-normal} displays the observed $\alpha$ levels. The parametric bootstrap (PB) method is accurate and the estimated $\alpha$ levels are close to the
 nominal level 0.05. The PT  test is liberal and its  approximated $p$-value is higher than the nominal level,  specially under exponential distribution.  \BL{We observe that the  WT test is a bit conservative under  the normal  and logistic distributions, i.e., the approximated $p$-values are lower than the nominal level}.  The observed $\alpha$ levels for  EAR follow the PB method closely and  they are less liberal than the PT under the exponential distribution. 

\begin{table}[htp]
\centering
\caption{Observed $\alpha$-levels of the proposed tests for testing  $H_0:\mu=0$ under the Normal distribution and $H_0: \mu=1$ for the Exponential and Logistic distributions.}\vspace{0.3cm}
%\begin{small}
\begin{tabular}{l  c c c c c c cc }
\hline\hline &&\multicolumn{1}{c}{ PT}&\multicolumn{1}{c}{ WT}&   \multicolumn{1}{c}{EAT} &  \multicolumn{1}{c}{EAR}  & \multicolumn{1}{c}{PB} & \\
%\hline & \multicolumn{5}{c}{\textbf{Normal Distribution}}\\
\hline\hline
 &$D_1$ &0.062&0.041& 0.056& 0.052&0.050  \\
&$D_2$ &0.078&0.039& 0.054& 0.056&0.054  \\
N(0, 1)&$D_3$ &0.072&0.038&  0.046& 0.047&0.049 \\
&$D_4$ &0.071&0.033&  0.057& 0.058&0.054\\
&$D_5$ &0.064&0.039&  0.043& 0.045&0.047 &\\
%\hline & \multicolumn{5}{c}{\textbf{Exponential Distribution}}\\
%\hline &\multicolumn{2}{c}{\qquad PT}&  \multicolumn{1}{c}{BT}  &  \multicolumn{1}{c}{EAT} &  \multicolumn{1}{c}{EAR} &  \multicolumn{1}{c}{PB}\\

\hline
&$D_1$ &0.107 &0.071& 0.081& 0.080&0.051  \\
&$D_2$ &0.133&0.072&  0.076& 0.079&0.049  \\
Exponential (1)&$D_3$ &0.132&0.081& 0.089& 0.090&0.054  \\
&$D_4$ &0.131&0.073&  0.098& 0.094&0.050  \\
&$D_5$ &0.098&0.074&  0.058& 0.055&0.053  \\

%\hline & \multicolumn{5}{c}{\textbf{Logistic Distribution}}\\
%\hline &\multicolumn{2}{c}{\qquad PT}&  \multicolumn{1}{c}{BT}  &  \multicolumn{1}{c}{EAT} &  \multicolumn{1}{c}{EAR} &  \multicolumn{1}{c}{PB}\\
\hline
&$D_1$ &0.052&0.042& 0.05& 0.051& 0.047  \\
&$D_2$ &0.076&0.041& 0.058&  0.059& 0.050  \\
Logistic (1, 1)&$D_3$ &0.065&0.033&  0.048& 0.050&0.046  \\
&$D_4$ &0.068&0.034&  0.059& 0.057& 0.051  \\
&$D_5$ &0.059&0.034&  0.043& 0.044& 0.041  \\

\hline
\end{tabular}
%\end{small}
\label{2observed_alpha-normal}
\end{table}

Table \ref{2shift-normal} displays  the estimated power values under shift alternatives $H_a: \mu=\mu_0+\delta$ with $\delta\neq 0$.  We used 95\% percentile bootstrap confidence intervals
for $\mu$, using  EAT and EAR to obtain the power of the test statistics at $\alpha=0.05$.  The entries of these tables are the proportion of times that the bootstrap confidence intervals do not cover zero.  \BL{Compared with PT, both the EAT and EAR methods lead to high powers}, hence they  can be nominated to conduct appropriate tests. The results of other simulation studies (not presented here) show similar behavior for other values of $k$ such as $k=2, 3, 8, 10$. We also considered different sample sizes. The better performance of the proposed methods are apparent for small and relatively small sample sizes (which often happens in practice for RSS) and they perform similarly when the sample size gets very large for a fixed set size.

\begin{table}
\centering
\caption{Power comparison for the proposed tests under location shift.}
 \scalebox{0.82}{
\begin{tabular}{l c c  ccc c  ccccc  ccccc  }\\
\hline\hline 
\multicolumn{2}{c}{} & \multicolumn{5}{c}{Normal dist.}& \multicolumn{5}{c}{Exponential  dist.} & \multicolumn{5}{c}{Logistic dist.}\\
\hline
$\delta$      &$D$  & \multicolumn{1}{c}{PT} & WT&  \multicolumn{1}{c}{ETA}& \multicolumn{1}{c}{ETR} & \multicolumn{1}{c}{PB}
                   & \multicolumn{1}{c}{PT} &WT&   \multicolumn{1}{c}{ETA}& \multicolumn{1}{c}{ETR} & \multicolumn{1}{c}{PB}
                    & \multicolumn{1}{c}{PT} &WT&   \multicolumn{1}{c}{ETA}& \multicolumn{1}{c}{ETR} & \multicolumn{1}{c}{PB}\\
\hline\hline 

0.1&$D_1$ & 0.148&0.097& 0.152& 0.145&  0.138 &0.229&0.148&   0.222& 0.208&   0.209 &0.076& 0.049& 0.088& 0.088&   0.077 \\
&$D_2$      &0.143&0.069&  0.140& 0.142&   0.139 &0.227&0.093&   0.216& 0.212&   0.208 &0.116&0.052&   0.118& 0.125&   0.112\\
&$D_3$ &0.145&0.061&   0.147& 0.150&  0.142 &0.255&0.130&   0.255& 0.241&   0.242 &0.122&0.037&   0.130& 0.128&  0.120\\
&$D_4$ &0.155&0.057&   0.156& 0.164&  0.149 &0.216&0.096&   0.216& 0.204&   0.205 &0.112&0.032&   0.118& 0.120&   0.116\\
&$D_5$ &0.141&0.064&   0.142& 0.141& 0.142 &0.190&0.143&   0.187& 0.176&  0.164 &0.108&0.034&   0.106& 0.102&   0.104\\
\hline

0.2&$D_1$ &0.389&0.297&   0.384& 0.388&  0.382 &0.416&0.304&   0.412& 0.388&   0.380 &0.162&0.102&  0.177& 0.184&   0.157 \\
&$D_2$ &0.340&0.185&   0.337& 0.344&  0.333 &0.375&0.180&   0.375& 0.359&   0.347 &0.175&0.085&   0.183& 0.191&   0.176 \\
&$D_3$ &0.333&0.143&   0.339& 0.335&  0.327  &0.405&0.235&   0.399& 0.385&   0.386 &0.159&0.057&  0.174& 0.175&   0.158\\
&$D_4$ &0.336&0.144&   0.336& 0.337&  0.336  &0.381&0.172&   0.379& 0.363&   0.360   &0.147&0.058&  0.155& 0.158&   0.147 \\
&$D_5$ &0.308&0.168&   0.310& 0.315&  0.312 &0.190&0.286&   0.187& 0.176&  0.164  &0.137&0.064&  0.141& 0.134&   0.139 \\
\hline

0.3&$D_1$ &0.696&0.600&   0.698& 0.684&  0.694 &0.644&0.500& 0.650& 0.618&   0.603 & 0.294&0.215& 0.291& 0.302&   0.282\\
&$D_2$ &0.571&0.351&   0.571& 0.569&  0.559 &0.553&0.292&   0.563& 0.538&   0.517 &0.258&0.145&  0.261& 0.261&   0.252\\
&$D_3$ &0.561&0.284&   0.564& 0.566&  0.549 &0.604&0.347&   0.598& 0.581&   0.568 &0.252&0.093&  0.264& 0.264&   0.249\\
&$D_4$ &0.569&0.302&   0.566& 0.565&  0.559 &0.524&0.281&   0.518& 0.520&   0.501 &0.223&0.102&   0.229& 0.232&   0.227\\
&$D_5$ &0.557&0.355&   0.549& 0.556&  0.541  &0.640&0.476&   0.621& 0.592&   0.573 &0.250&0.129&  0.251& 0.252&   0.243\\
  \hline
\end{tabular}}
\label{2shift-normal}
\end{table}

  \begin{table}
  \centering
  \caption{Observed $\alpha$-levels for the proposed tests  for testing  $H_0:\mu=0$ for normal distribution and  $H_0:\mu=1$ for the exponential and logistic distributions, under imperfect ranking.}
  \vspace{0.3cm}
  \scalebox{0.8}{
  \begin{tabular}{c  c c ccc  ccccc }
  \hline\hline & \multicolumn{5}{c}{$\sigma_\epsilon=0.5$}& \multicolumn{5}{c}{$\sigma_\epsilon=1$}\\
  \hline 
  $D$     &  PT  &ETA&ETR & IETA&IETR  &  PT  &ETA&ETR & IETA&IETR\\
  \hline \hline
  & \multicolumn{10}{c}{\textbf{Normal Distribution}}\\
  \hline
  
      $D_1$ & 0.056& 0.054& 0.054& 0.053& 0.056&0.069& 0.068& 0.066& 0.066& 0.068\\
     $D_2$ & 0.072& 0.072& 0.070& 0.070& 0.073& 0.074& 0.077& 0.081& 0.071& 0.077   \\
     $D_3$ &0.067& 0.066& 0.069& 0.067& 0.067& 0.087& 0.081& 0.079& 0.081 &0.077 \\
     $D_4$ & 0.058& 0.057& 0.057& 0.060& 0.056 & 0.068 &0.070& 0.066& 0.060 &0.066 \\
     $D_5$ &0.067& 0.063& 0.067& 0.066& 0.065&0.067& 0.070& 0.069& 0.069& 0.066 \\
     \hline & \multicolumn{10}{c}{\textbf{Exponential Distribution}}\\
     \hline 
 
      $D_1$ &0.073& 0.065& 0.068 &0.068& 0.068 & 0.067 &0.059 &0.060 &0.059& 0.056\\
   $D_2$ & 0.084 &0.076& 0.079& 0.077& 0.078&0.083 &0.078 &0.075 &0.078 &0.074\\
   $D_3$ & 0.099 &0.094 &0.094& 0.093& 0.092&0.063& 0.058& 0.063 &0.053& 0.051 \\
   $D_4$ & 0.103& 0.100& 0.099& 0.096 &0.097 &0.076 &0.082 &0.076 &0.068 &0.070\\
   $D_5$ &  0.078& 0.069& 0.070 &0.069& 0.069 &0.071& 0.067 &0.066 &0.059& 0.065\\ 
     \hline & \multicolumn{10}{c}{\textbf{Logistic  Distribution}}\\
     \hline  
   $D_1$ & 0.060& 0.061& 0.061& 0.062& 0.064& 0.058& 0.061 &0.061& 0.056 &0.061\\
  $D_2$ &   0.071& 0.074& 0.074& 0.070& 0.073& 0.075& 0.076 &0.079& 0.076 &0.079 \\
  $D_3$ &  0.077& 0.078& 0.079& 0.081& 0.079&0.071 &0.072 &0.072 &0.067 &0.071 \\
  $D_4$ &  0.078 &0.080 &0.080& 0.080& 0.078 &0.075& 0.079 &0.075 &0.075& 0.076 \\
  $D_5$ &   0.068& 0.069& 0.065 &0.067 &0.067&0.064& 0.060& 0.063& 0.064 &0.063\\
  \hline 
  \end{tabular}
 }
  \label{3observed_alpha-normal}
  \end{table}

\begin{table}
\centering \footnotesize 
\caption{Power comparison for the proposed tests under location shift and imperfect ranking with $\sigma_\epsilon=0.5$.}
 \scalebox{0.82}{
\begin{tabular}{cc   ccccc   ccccc  ccccc  }\\
\hline\hline 
\multicolumn{2}{c}{} & \multicolumn{5}{c}{Normal dist.}& \multicolumn{5}{c}{Exponential  dist.} & \multicolumn{5}{c}{Logistic dist.}\\
\hline
$\delta$      &$D$  &  PT  &EAT&EAR & IEAT&IEAR&  PT  &EAT&EAR & IEAT&IEAR&  PT  &EAT&EAR & IEAT&IEAR  \\
\hline\hline
 0.1&$D_1$&0.162& 0.162& 0.163& 0.168& 0.162& 0.218& 0.212& 0.203& 0.208& 0.202&0.090& 0.099& 0.101 &0.100& 0.105\\
&$D_2$&0.163 &0.161& 0.168& 0.170& 0.174& 0.211& 0.208 &0.201 &0.199 &0.195& 0.105& 0.114& 0.110& 0.112& 0.109\\
&$D_3$&0.143 &0.146 &0.152& 0.143 &0.150 &0.230 &0.223& 0.221& 0.225& 0.225&0.105 &0.116 &0.116& 0.113& 0.119\\
&$D_4$&0.149& 0.155& 0.154& 0.154 &0.161&0.222 &0.215& 0.212 &0.210& 0.208&0.105& 0.111& 0.116& 0.112& 0.114  \\
&$D_5$&0.158& 0.158 &0.159 &0.155& 0.155& 0.212& 0.190& 0.189 &0.193& 0.190&0.088& 0.090 &0.093& 0.094& 0.094 \\

\hline
0.2&$D_1$&0.394& 0.397& 0.399& 0.404& 0.40&0.413& 0.382& 0.381& 0.388& 0.381&0.160 &0.171& 0.169 &0.171& 0.166\\
&$D_2$ &0.349& 0.353 &0.355 &0.357 &0.358&0.379& 0.362& 0.364& 0.355 &0.350&0.159& 0.171& 0.172 &0.169& 0.172\\
&$D_3$&0.325& 0.340 &0.337 &0.333 &0.340&0.394 &0.373 &0.378& 0.374 &0.375&0.152 &0.157& 0.161& 0.156& 0.162 \\
&$D_4$ &0.332& 0.327& 0.333 &0.328 &0.338&0.373 &0.354& 0.358 &0.352 &0.351&0.169& 0.171 &0.176& 0.179 &0.181\\
&$D_5$  &0.326& 0.322 &0.328& 0.328 &0.334&0.412 &0.374 &0.374 &0.372& 0.367&0.148 &0.156& 0.151 &0.152& 0.153\\
\hline

0.3&$D_1$&0.709& 0.708 &0.704 &0.709 &0.706 &0.643 &0.615 &0.609 &0.607 &0.607&0.303& 0.310& 0.309& 0.309& 0.312\\
&$D_2$ &0.584 &0.588& 0.586& 0.584 &0.588&0.517 &0.498 &0.498 &0.493 &0.482&0.259& 0.269& 0.273 &0.277& 0.275\\
&$D_3$ &0.556 &0.563 &0.561 &0.557 &0.556&0.594 &0.571 &0.570 &0.569 &0.566&0.238 &0.249 &0.254 &0.250 &0.255\\
&$D_4$ &0.571 &0.570& 0.565 &0.565 &0.569&0.530 &0.516 &0.518 &0.508 &0.507&0.247 &0.249 &0.249 &0.248 &0.257 \\
&$D_5$ &0.558 &0.563 &0.555& 0.556 &0.561&0.619 &0.572 &0.574 &0.568 &0.563&0.247 &0.255& 0.252& 0.255 &0.251\\
  \hline
\end{tabular}}
\label{imperfectpower}
\end{table}

%%%%%%%%%%%%%%%%%%%%%%%%%%%%%%%%%%%%%%%%%%%%%%%%%%%
%%
%%         Subsection: Imperfect ranking
%%
%%%%%%%%%%%%%%%%%%%%%%%%%%%%%%%%%%%%%%%%%%%%%%%%%%%%

\subsection{Imperfect ranking}\label{imperfect}
\BL{In this section, we compare the finite sample performance of our proposed bootstrapping techniques with the PB   under imperfect ranking cases. 
In order to produce the imperfect URSS/RSS samples, we use the model proposed by  Dell and Clutter (1972). Let $X_{[i]j}$ and $X_{(i)j}$ denote the judgment and true order statistics, respectively.} Suppose
\[
X_{[i]j}=X_{(i)j}+\epsilon_{ij},~~\epsilon_{ij}\sim N(0,\sigma_\epsilon),
\]
where $X_{(i)j}$ and $\epsilon_{ij}$ are independent. 

Using imperfect URSS with $\sigma_\epsilon=0.5$ and 1,  we report the observed significance levels for testing $H_0: \mu=\mu_0$ against $H_a: \mu>\mu_0$ for different methods in Table \ref{3observed_alpha-normal}. These choices of $\sigma_{\epsilon}$ resulted in the  observed correlation coefficients of 0.89 and 0.70 between the ranking variable and the variable of interest, respectively. \BL{As compared with  the results under the  perfect ranking assumption, 
 the proposed methods seem to be  robust with respect to imperfect ranking. It was shown that the test under exponential distribution for the imperfect sampling is a bit liberal.  We also observe that imperfect ranking affects   the power of the tests since, as it is shown in Table \ref{imperfectpower},   by adding errors in ranking, the power of the proposed tests  decreases.  
 The importance of accurate ranking in RSS designs has been mentioned in several works. Frey, Ozturk and Deshpande (2007) considered nonparametric tests for the perfect judgment ranking. Li and Balakrishnan (2008) proposed several nonparametric tests to investigate  perfect ranking assumption. Vock and Balakrishnan (2011) suggested a Jonckheere-Terpstra type test statistic for perfect ranking in balanced RSS. These tests are further studied by Frey and Wang (2013) and compared with the most powerful test. 
 }

  % % % % % % % % % % % % % % %
 % % % % % % % % % % % % % % % % % % % %
\BL{ 
 In order to derive  the theoretical  results under the imperfect ranking  assumption, one can proceed as follow. First, note that under the imperfect ranking the density function of characteristic of interest for the unit judged to be  ranked $r$ is no longer $f_{(r)}$. We  denote this density with $f_{[r]}$. One approach to derive the CDF  $F_{[r]}$ of the $r$th judgmental order statistic is  to use the following model
 \begin{eqnarray}
 F_{[r]}=\sum _{s=1}^k p_{sr} F_{(s)}(x), \label{emp2}
 \end{eqnarray}
 where $p_{sr}$ is the probability that the $s$th order statistic is judged to have  rank $r$, with  $\sum_{s=1}^{k}p_{sr}=\sum_{k=1}^{k}p_{sr}=1$. 
 %%%%
 \begin{lemma} Suppose  the imperfect ranking in the RSS design is such that 
 \[
 F_{[r]}(x)=\sum_{s=1}^{k}p_{sr}F_{(s)}(x),\quad \forall x\in\mathbb{R}. 
 \]
  For the  resampling technique EAT (or EAR)  under the imperfect ranking assumption, which is denoted by IEAR (or IEAR), we have  
  \[
  \underset{t\in \mathbb{R}}{\sup}|\widehat F^{*}_{<n>}(t)-F(t)|=0,
  \]
  where $\widehat F^{*}_{<n>}(t)$ is the EDF  of the resulting   bootstrap sample. 
 \end{lemma}
 %%%
  \begin{proof}
  We first note that using the IEAT (or IEAR), we have    
 \[
 F^*_{[r]}(t)=\sum_{s=1}^{k}p_{sr}F^*_{(s)}(t),
 \] 
 where $\tilde F^*_{[r]}(.)$ and $\tilde F^*_{(r)}(t)$ are the EDF of the resulting bootstrap samples under the  IEAT (or IEAR)  and EAT (or EAR), respectively. One can easily show that
 \begin{eqnarray}
 \tilde {\widehat F}^{*}_{<n>}(t)=\frac{1}{k}\sum_{r=1}^{k} \tilde F^*_{[r]}(t)= \frac{1}{k}\sum_{r=1}^{k} \sum_{s=1}^{k}p_{sr}\tilde F^*_{(s)}= \frac{1}{k}
  \sum_{s=1}^{k} \sum_{r=1}^{k} (p_{sr}) \tilde F^*_{(s)}(t)=
  \frac{1}{k}
   \sum_{s=1}^{k}  \tilde F^*_{(s)}(t)= \tilde {\widehat F}_n^*(t), \nonumber
 \end{eqnarray}
 Hence,  we have
  \[
 \tilde {\widehat F}^{*}_{<n>}(t)-F(t)= (\tilde {\widehat F}^{*}_{<n>}(t)-\tilde {\widehat F}(t))+(\tilde {\widehat F}^{*}_{n}(t)-F(t))=O(\frac{1}{mk}),
  \] 
  and this completes the proof.
 \end{proof} 
}

 %%%%%%%%%%%%%%%%%%%%%%%%%%%%%%%
 %
 %
 %
 %%%%%%%%%%%%%%%%%%%%%%%%%%%%%%%%
 
 \subsection{Comparison with the empirical likelihood method }
In  this section,  we compare the  performance of the bootstrap tests based on ET estimators  of $F$  with the  one based on the empirical likelihood estimator of $F$ which is already studied in the literature by  Baklizi (2009) and Liu et al. (2009). Empirical likelihood is an estimation method based on likelihood functions without having to specify a parametric family for the observed data.  Empirical
likelihood methodology has become a powerful and widely applicable tool for non-parametric  statistical inference and it has been  used under different sampling designs. For a comprehensive review of the  empirical likelihood  method and some of its variations see Owen (2001).  
For testing the null hypothesis $H_0:\mu=\mu_0$ using the empirical likelihood estimator of $F$ based on a balanced RSS sample, 
 Baklizi (2009) showed that under the  finite variance assumption
 \begin{eqnarray}
 C_0\, l(\mu_0)\overset{L}{\rightarrow}\chi^2_1,
 \end{eqnarray} 
 where 
 \begin{eqnarray}
  C_0= \frac{\sum_{r=1}^k \sigma^2_r+\sum(\bar{X}_{[r]}-\mu_0)^2}{\sum_{r=1}^k \sigma^2_r}\quad\text{and}\quad
  l(\mu_0)= 
   \frac{\{\sum_{r=1}^k\sum_{j=1}^{m} (X_{[r]j}-\mu_0)\}^2}{\sum_{r=1}^k\sum_{j=1}^{m} (X_{[r]j}-\mu_0)^2}.
  \end{eqnarray} 
However, this is a liberal test for small samples and it does not work for   URSS case.   Liu et al. (2009) proposed to use  the empirical likelihood  method for RSS data  by first  averaging  the observations of each cycle to construct  
  \[ \bar{\mathcal{X}}_j=\frac{1}{m} \sum_{r=1}^k X_{(r)j}, ~ j=1,\ldots,m. \]
  Then, by observing  that $\bar{\mathcal{X}}_j$ are i.i.d.\ samples from $F$,   Liu et al.\ (2009) constructed  the usual  empirical likelihood estimator of $F$ and used it for a testing hypothesis problem. As we show below this method does not perform well,  especially  for   RSS samples when the  number of cycles is small.   
  
   The following simulation study shows that  using  EAR  based on the ET estimator of $F$ can be used to overcome these difficulties. To this end, we consider a balanced RSS with small sample, i.e., $D_6=(2,2,2,2,2)$.  Figure \ref{qqp}, shows the  Q-Q plots of the  $p$-values  based on the EAR algorithm  (first column), and those  proposed by Baklizi (2009)  (second column) and Liu et al.\ (2009) (the third column), respectively for the normal distribution when $H_0: \mu=0$ and the exponential and logistic distributions for $H_0: \mu=1$.

\begin{figure}[htb!]\centering
\includegraphics[scale=.7]{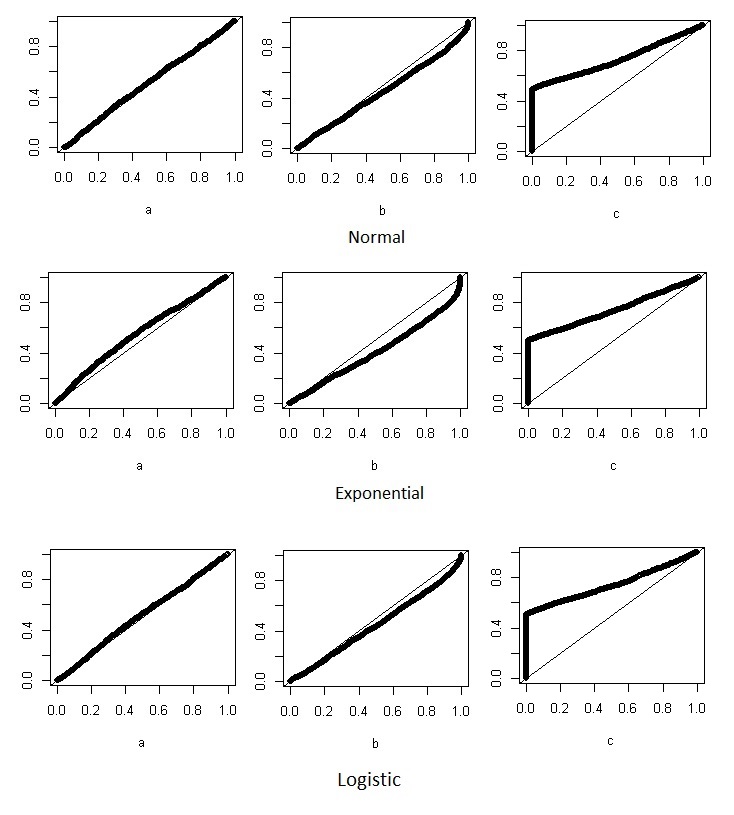}\caption{The Q-Q plot  for the p-values of the proposed  test statistic based on the EAR algorithm  (first column), and those  proposed by Baklizi (2009)  (second column) and Liu et al. (2009) (the third column), respectively for the normal distribution when $H_0: \mu=0$ and the exponential and logistic distributions for $H_0: \mu=1$.  
}
\label{qqp}
\end{figure}

 %%%%%%%%%%%%%%%%%%%%%%%%%%%%%%%%%%%%%%%%%%%%%%%%%%%
%%
%%         Subsection 4: Concluding Remarks

%%
%%%%%%%%%%%%%%%%%%%%%%%%%%%%%%%%%%%%%%%%%%%%%%%%%%%%
\section{Real data application}\label{real data}
In this section,  we  use a data set  containing the birth weight and seven-month weight of 224 lambs along with
the mother's weight at time of mating, collected at the Research Farm of Ataturk University, Erzurum, Turkey.    Jafari Jozani and Johnson (2012) as well as Ozturk and Jafari Jozani (2014)   used this data set to study the performance of ranked set sampling in estimating the  mean,  the  total values and  quantiles   of the seven-month weight of these lambs. The measurement of  the weight of young sheep is usually  labor intensive due to their active nature, and measurement errors can be inflated due to this activity.  However, one can easily rank a small number of lambs based on their birth weights  or their mother's weights to perform a ranked set sampling design hoping that the RSS sample results in a more representative sample from the whole population. 
Here, we treat these 224 records as our population, with the goal of  a testing hypothesis problem about the mean of  the  weight distribution  of these 224 lambs at seven-month. We consider both perfect and imperfect ranking cases. For the  perfect ranking scenario, ranking is done based on the weight of lambs at seven-month. For the  imperfect ranking,  we consider two cases.  In the first case (Imperfect 1),  ranking is done based on the   the birth weight of the lambs.   The Kendall's $\tau$ between the seven-month weight  and the birth weight is  0.64. In the second case (Imperfect 2),  we  perform the  ranking process  based on the mother's weight at time of mating which results in a small  Kendal's $\tau$ of  0.41  between  the lambs weight at  seven-month  and mother's weight at the time of mating.
Summary statistics for  these variables for  the underlying population are presented in Table \ref{tbl:description}.   Figure \ref{yeild} shows the histogram of the  seven-month weight of  these lambs with a kernel density estimator of their weight distribution. We also present the scatter plots of the  birth weight and mother's weight of these lambs against  their weight at seven-months. We observe that there is a stronger association between the seven-month weight and   the birth weight of these lambs. So, we expect to observe a better results under the Imperfect 1 scenario.

\begin{table}[htdp]\small
\caption{Summary statistics for the values of the birth weight and seven-month weight of 224 lambs along with
the mother's weight at time of mating, collected at the Research Farm of Ataturk University, Erzurum, Turkey}
\begin{center}
\begin{tabular}{cccccccc}\hline\hline 
                            Variable  &    Min & $Q_1$ &  Median& Mean & $Q_3$ & Max &  $\sigma^2$ \\ \hline 
    Seven-month weight      &     20.30  & 25.50& 27.90 &  28.11 & 31.00& 40.50 &15.21 \\ \hline
                  Birth weight      &     2.50  & 3.87 & 4.40 &4.36& 4.80 & 6.70 & 0.63  \\ \hline
             Mother's weight    &    42.20 & 49.68&52.30&52.26& 55.10&63.70 &19.22  \\ \hline\hline
  \end{tabular}
\end{center}
\label{tbl:description}
\end{table}%

Table \ref{2observed_alpha-real2} presents the results of the analysis for a testing hypothesis problem to test $H_0: \mu=28.11$ based on different RSS sampling designs as in Section \ref{simulation}. Based on the  obtained  $\alpha$-level for each sampling design under the PT and EAR algorithm we observe that our  proposed bootstrap test using the ET estimator of the  DF  shows a satisfactory performance compared with the PT method in both perfect and imperfect ranking scenarios.

\begin{figure}[htb!]\centering
\includegraphics[width=7in]{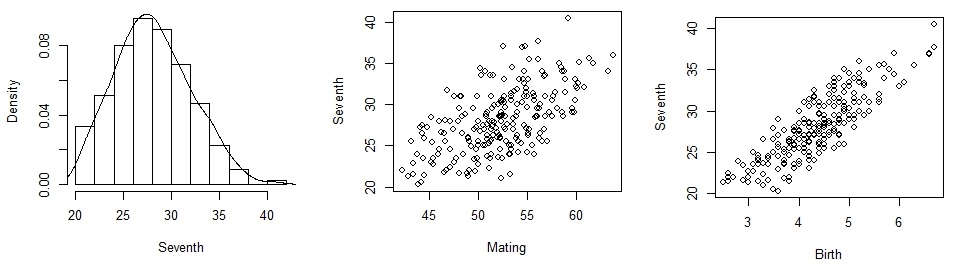}\caption{The histogram of the  values of  seven-month weight of  224  lambs with a kernel density estimator of their weight distribution as well as   the scatter plots of the  birth weight and mother's weight of these lambs  against  their weight at seven-months.}
\label{yeild}
\end{figure}

\begin{table}[h!]
\centering
\caption{The values  of  the observed $\alpha$-levels for testing  $H_0:\mu=28.11$ for the weight distribution of a population of 224 lambs   based on different perfect and imperfect  RSS design using the PT and EAR algorithm.}\vspace{0.3cm}
\begin{small}
\begin{tabular}{c c c c c c c }
\hline \hline
& Method & $D_1 $  & $D_2$ & $D_3$ & $D_4$ & $D_5$\\
\hline \hline
Perfect Ranking &  PT &0.062&0.094&0.085&0.076&0.083  \\
  &EAR &0.055&0.052&0.044&0.046&0.047 \\
\hline 
Imperfect 1&   PT & 0.064&0.082&0.091&0.087&0.094 \\
  &EAR &0.048& 0.047& 0.052& 0.047& 0.045\\
 \hline
  Imperfect 2&  PT &0.065&0.090&0.086&0.086&0.091 \\
  &EAR &0.048&0.042&0.051&0.043&0.046\\
 \hline\hline
\end{tabular}
\label{2observed_alpha-real2}
 \end{small}
\end{table}

\section{Concluding Remarks}\label{concluding}
We propose nonparametric estimators of the cumulative distribution of  a continuous random variable  using the ET empirical likelihood method based on ranked set sampling designs. The ET DF estimators are used to  construct new  resampling  techniques for URSS data.  We study different   properties of the proposed algorithms. For a hypothesis  testing  problem, we show that  the bootstrap test based  on exponential tilted estimators exhibit a small bias of order $O(n^{-1})$, which is  a very desirable property.   We compared the performance of our proposed techniques with those based on empirical likelihood. The latter are developed under the balanced RSS assumption and they are not applicable for URSS situation.  The results of  the  simulation studies  as well as a real data application show that the method based on ET estimators of the DF perform very well even for moderate or small sample sizes.

%%%%%%%%%%%%%%%%%%%%%%%%%%%%%%%%%%%%%%%%%%%%%%%%%%%
%%
%%         Section : Appendix
%%
%%%%%%%%%%%%%%%%%%%%%%%%%%%%%%%%%%%%%%%%%%%%%%%%%%%%

\section*{Acknowledgements}
\BL{ We gratefully acknowledge the constructive comments of the referees and the associate editor}. The research of  M.\ Jafari Jozani was  supported by the NSERC of Canada. The research of R.\ Modarres was supported in part by
the National Institute of Health, under the Grant No.\ 1R01GM092963-01A1.

%\begin{table}[h]
%\centering
%   \caption{Variance estimates using the proposed algorithms}
%\begin{tabular}{lcccccc} \hline
%&&\multicolumn{5}{c}{Design}\\ \cline{3-7}
%Distributions&test&$D_1$ &$D_2$&$D_3$&$D_4$&$D_5$\\ \hline
%$X\stackrel{}{\sim} N(0,1)$
% &$\widehat \sigma^2$&0.014&0.022&0.021&0.021&0.020\\
% &$S^2$&0.014&0.022&0.021&0.021&0.021\\
% &BRSSR&0.011&0.013&0.014&0.014&0.013\\
% &BRSS&0.007&0.009&0.009&0.009&0.008\\
% &MRSRSS&0.014&0.026&0.019&0.018&0.027\\
% &MBRSS&0.014&0.023&0.021&0.020&0.022\\
% &BT&0.014&0.023&0.021&0.021&0.022\\
%  &ETALL&0.014&0.027&0.020&0.019&0.027\\
%  &EAR&0.014&0.024&0.021&0.020&0.020\\
% \hline
%$X\stackrel{}{\sim} exp(1)$
% &$\widehat \sigma^2$&0.018&0.032&0.026&0.028&0.018\\
% &$S^2$&0.018&0.032&0.026&0.029&0.018\\
% &BRSSR&0.014&0.020&0.017&0.019&0.013\\
% &BRSS&0.011&0.015&0.013&0.015&0.009\\
% &MRSRSS&0.018&0.027&0.020&0.024&0.029\\
% &MBRSS&0.017&0.031&0.024&0.028&0.019\\
% &BT&0.018&0.032&0.024&0.028&0.019\\
%  &ETALL&0.017&0.037&0.019&0.030&0.019\\
%   &EAR&0.018&0.030&0.024&0.028&0.019\\
% \hline
%$X\stackrel{}{\sim} Lgistic(1,1)$
% &$\widehat \sigma^2$&0.051&0.072&0.080&0.076&0.069\\
% &$S^2$&0.051&0.075&0.078&0.076&0.069\\
% &BRSSR&0.040&0.046&0.051&0.051&0.046\\
% &BRSS&0.028&0.032&0.037&0.037&0.031\\
% &MRSRSS&0.050&0.087&0.066&0.063&0.091\\
% &MBRSS&0.050&0.076&0.076&0.072&0.071\\
% &BT&0.051&0.077&0.076&0.073&0.072\\
%  &ETALL&0.051&0.090&0.070&0.068&0.092\\
%   &EAR&0.051&0.073&0.077&0.075&0.073\\
%\hline
%\end{tabular}\\
%\label{tab2}
%\end{table}

 \newpage

\end{document}